\documentclass{article}

\usepackage{algorithm}
\usepackage{algorithmic}
\usepackage[T1]{fontenc}
\usepackage[utf8]{inputenc}
\usepackage{lmodern}
\usepackage{amsfonts}
\usepackage{amsmath}
\usepackage{amssymb}
\usepackage[standard]{ntheorem}
\usepackage[english]{babel}
\usepackage{url}
\usepackage{graphicx}
\usepackage{color}
\usepackage{subfigure}
\usepackage{listings}

\begin{document}

\title{Epidemiological Approach for Data Survivability in Unattended Wireless Sensor Networks}

\author{Jacques M. Bahi, Christophe Guyeux, Mourad Hakem, and Abdallah Makhoul}

\maketitle

\begin{abstract}
Unattended Wireless Sensor Networks (UWSNs) are Wireless Sensor Networks characterized by sporadic sink presence and operation in hostile settings. The absence of the sink for period of time, prevents sensor nodes to offload data in real time and offer greatly increased opportunities for attacks resulting in erasure, modification, or disclosure of sensor-collected data. In this paper, we focus on UWSNs where sensor nodes collect and store data locally and try to upload all the information once the sink becomes available. 
One of the most relevant issues pertaining UWSNs is to guarantee a certain level of information survivability in an unreliable network and even in presence of a powerful attackers. 
In this paper, we first introduce an epidemic-domain inspired approach to model the information survivability in
UWSN. Next, we derive a fully distributed algorithm that supports these
models and 
give the correctness proofs.

\end{abstract}

\section{Introduction}

Unattended Wireless Sensor Networks (UWSNs), which have been introduced by Di Pietro \emph{et al.} in~\cite{DiPietro08}, are WSNs characterized by the sporadic presence of the sink. These UWSNs are useful for instance to detect poaching in a national park, or as a monitoring system to check the pressure of an underground pipeline, as
stated in ~\cite{DiPietro13}. In such networks, nodes collect data from the area under consideration, and then they try to upload all the stored data when the sink comes around. It is motivated by the scenarios where only historical information or digest data, not real-time data, are of interest. For example, what is average temperature during last three months; what is the highest and lowest humidity degree during last 24 hours; or more specifically, what is average content of a chemical element in soil during last half year~\cite{9}. 

Due to the absence of a direct and alive connection with the
sink, these networks are more subject to malicious attacks
than traditional WSNs. Sensor nodes may malfunction due
to some threats, e.g., physical failure such as melting, corroding
or getting smashed, and more sophisticated, mobile
adversary attacks. Some data may be lost, erased or
modified before the arrival of the mobile collectors, which
significantly affects functionalities of UWSNs. Therefore,
the critical issue for UWSNs is how to maximize information
survivability~\cite{Ma08}: the dimension of the area is often prohibitive in such networks, while the absence of the sink facilitates the work of attackers. Data or information survivability consists on preserving data for a long period of time in the face of attacks, which is crucial for designing safe UWSNs.

To ensure data survivability, cryptographic or noncryptographic
approaches can be used. End-to-end encryption schemes that support operations over cypher-text
have been proved important for private party sensor network
implementations or other security schems~\cite{bgm11:ij, 10, 11}. Unfortunately, nowadays these methods are
very complex and not completely suitable for sensor nodes having limited
resources. Non-cryptographic approaches are suitable for low-cost sensors,
that do not have the capability to execute computational
intensive calculation. In this paper, we will focus on non-cryptographic
approaches for data survivability. We propose an epidemic-domain inspired approach to model the information survivability in UWSN.

The model we present here is based on both SIR (Susceptible - Infected - Recovered) and SIS (Susceptible - Infected - Susceptible) models. A node is susceptible to a data item when it is online and functioning normally; he can receive the information that must survive. Intuitively, the model we
focus on resembles an SIR model studied previously in~\cite{DiPietro11,DiPietro13}. Our novelty is that we study arbitrary dynamic network topologies instead of static networks.
In a next step, 
we provide a fully distributed algorithm which supports/covers different epidemic models. The aim of this algorithm is to ensure data survivability in UWSN by maintaining a subset of safe nodes in working state while replacing/locking the attacked ones when needed. 

The remainder of the paper is organized as follows: Section~\ref{PW} briefly reviews the related work. The SIR model for Data Survivability in UWSNs is presented in Section~\ref{SIR}. Section~\ref{USIR} details the proposed epidemic schemes in a comparatively manner.
We present in section~\ref{sec:algo} the design and analysis of the proposed epidemic algorithm and give the proofs.
The next section is devoted to numerical simulations.
Finally, Section~\ref{CONC} concludes this research work.

\section{Previous work}
\label{PW}

Wireless sensor networks have attracted a lot of interest over the last decade, for quick and efficient aggregation of information~\cite{1,2}, for understanding "trust" and "distrust" in online networks~\cite{bgmp12:ij, bgm11:ij, 3}, and in several other areas. In most previous researches on WSNs, sensed data are assumed to be collected in a real-time manner and sensor nodes always connected to the sink that is generally the only unconditionally trusted entity. In order to extending the network lifetime a possible solution presented in the literature and consists on using mobile sinks or robots to perform several tasks instead of the nodes themselves. For instance, the authors in~\cite{Mak2007, Mak08} propose a mobile beacon based approach to localize sensor nodes and ensure the network coverage. In~\cite{12}, a distributed algorithm is proposed in order to reduce the overhead message by using only local information and assisted by mobile sinks. Recently, a number of approaches exploiting sink mobility for data
collection in WSNs have been proposed~\cite{13, 14, 15, 16}. The main objective of these work is to reduce the energy consumption by optimizing the number of communications between nodes. In these approaches only single hop communication is required between nodes and the mobile sinks. Otherwise, in some cases and scenarios sensor nodes can not be connected for a short or long period of time and the network is left unattended while sensor nodes cannot offload data in real time. Therefore, sensor nodes wait for the new passing of a mobile sink for examle to send their data. Unattended Wireless Sensor Networks (UWSNs), have been introduced by Di Pietro \emph{et al.} in~\cite{DiPietro08}, where adversaries can compromise some sensor nodes and selectively destroy data. Other works have studied this problem in order to prevent such attacks, and ensure the "data survivability" in UWSN. 

The authors in~\cite{4} proposed a method to face an adversary that
indiscriminately erases all sensor data, and then in~\cite{5} cryptographic
techniques that prevent the adversary from recognizing
data that it aims to erase have been introduced. Sensor
cooperation to achieve self-healing in stationary UWSNs
has been explored in~\cite{Ma08}. \cite{6} proposes a new strategy based on the concept of secure multi-party protocols. The main advantages of the proposed strategy are not limited to the security, since it preserves privacy, enables the use of data aggregation and enforces a level of trust among nodes, which collaborate to compute aggregation functions. In~\cite{7}, the authors propose a dependable and efficient data survival scheme to maximize the data survival degree upon data retrieval. This technique makes use of computational secret sharing to achieve fault tolerance and compromise resilience, and uses network coding to further improve communication efficiency. In~\cite{8}, the authors focus on the conditions under which a sensor node can survive in an unreliable network. They propose and solve the problem using non-linear dynamical systems and fixed point stability theorems.

The epidemiology community has developed
the so-called SIR and SIS models~\cite{DiPietro11,DiPietro13} of infection. The
SIS model (Susceptible - Infected - Susceptible) is suitable
for, e.g., the common flu, where nodes may be infected, healed
(and susceptible), and infected again. The SIR model (Susceptible
- Infected - Recovered) is suitable for, say, mumps, where
a node, after being infected, becomes recovered (with life-time
immunity). SIS, SIR, and SIRS models have been investigated by authors of these
research works, in order to derive the parameters that can assure
information to survive. In these articles,
$S(t)$ compartment is constituted by sensors that do not possess 
the datum at time $t$, while $I(t)$ is the compartment of sensors that
possess it. Finally, the $R(t)$ compartment is constituted by sensors 
that have been compromised by the attacker.

On the other side and surprisingly, the authors in~\cite{DiPietro11,DiPietro13} never
consider that in a wireless sensor network, nodes' energy is provided
by a battery that can be emptied due to data acquisition, transmission,
or simply functioning cost of keeping alive. More precisely, the topology
of the networks they consider is static, the network's lifetime is unbounded,
and sensors cannot die due to empty batteries. Indeed, their work is more
related to unattended wired sensor networks, on main power but not with a 
battery as $S+I$ (SIS model) or $S+I+R$ (SIR and SIRS models) are 
constant. Our intention is to deepen their interesting work, by bringing their
proposal from wired sensor networks to WSNs, refining theirs models, and 
producing more theoretical results on each model. 

\section{A SIR model for Data Survivability in UWSNs}
\label{SIR}



\subsection{Introducing the Kermack \& McKendrick model}

In this section, the SIR model formerly presented in~\cite{DiPietro11,DiPietro13} 
is firstly recalled. Then, consumption hypotheses underlined in this model
are precised while theoretical results on the behavior of the compartments
of the network are further investigated.

In unattended wireless sensor networks the presence
of the sink is  sporadic. However the duration between two visits of the
sink to the network (its absence) can sometimes be considered negligible, 
in a first approximation,
compared to the time required to empty a sensor battery. 
In such UWSNs, the death processes of sensors can be neglected 
 if the aim is to study the immediate consequences of an attack between
 two visits of the sink. Under such an assumption, the global network can be divided in three
compartments, namely the sensors  $S$ susceptible to receive the datum of
interest (intrusion detection, etc.), the ones that currently store it $I$, 
and the recovered sensors $R$ that have been 
compromised by the attacker: their stored datum has been recovered.

Suppose now that between $S$ and $I$, the transmission rate is $b I$, where
$b$ is the contact rate, which is the probability of transferring the information in a
contact between a susceptible sensor and a sensor having the datum. Indeed, as proven
by Di Pietro \emph{et al.}, such a situation occurs when then
wireless sensor network is composed
by $n$ sensor, and if each sensor forwards the datum with
probability $\dfrac{\alpha}{n}$~\cite{DiPietro11,DiPietro13} ($\alpha$ is the transition rate). 

Suppose additionally that the rate to pass between $I$ and $R$, is $c$:
the attacker is able to individuate the sensors containing the target
information, and to destroy each of them with this probability $c$.
Notice that, if the duration of the information survivability
is $D$, then $c = \frac{1}{D}$, as a sensor experiences one recovery
in $D$ units of time.

\begin{figure}[ht]
\centering
\includegraphics{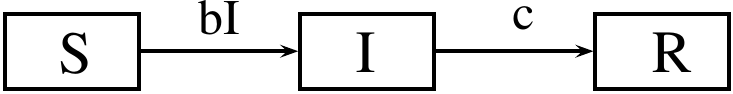}
\caption{SIR model}
\label{SIRmodel}
\end{figure}

Under such hypotheses and as stated in~\cite{DiPietro11,DiPietro13}, 
the sensors population follows the so-called SIR model 
of Kermack \& McKendrick~\cite{Kermack27} depicted in Figure~\ref{SIRmodel}.
Remark that the total sensors population is equal to $N=S+I+R=S_0+I_0+R_0$, 
which is a constant: the number of awaken, alive sensors does not evolve. 
In particular, only two of the three populations of sensors have to be studied.

\subsection{Firsts theoretical results}

Consider now that $x(t)=\dfrac{X(t)}{N}$ denotes the fraction of individuals
in the compartment $X$.
The SIR model can be expressed by the following set of ordinary non-linear differential equations: 
\begin{equation}
\label{modelSir1}
\left\{
\begin{array}{l}
\dfrac{ds}{dt} = - b i s\\\\
\dfrac{di}{dt} = b i s - c i\\\\
\dfrac{dr}{dt} = c i .\\
\end{array}
\right.
\end{equation}

Obviously, the typical time between transmissions is $T_t=b^{-1}$ while the 
typical time until attack when having the information is equal to
$T_e=c^{-1}$. Thus 
$$\dfrac{T_t}{T_e} = \dfrac{c}{b}$$
is the average number of transmissions between a sensor having the datum and
others before it lost this information due to the attacker.
Such a statement explain why, in the SIR historical model, the dynamics of the infectious 
class depends on the \emph{reproduction ratio} defined by 
$$R_0=\dfrac{b}{c},$$
which corresponds here to the expected number of new 
informed sensors (so-called ``secondary infections'')
providing a single sensor with the datum where all sensors are susceptible.
Furthermore, direct standard analysis manipulations (variables separation and then integration) lead 
to the following form for the susceptible sensors compartment: $s(t)=s(0) exp\left(-R_0 (r(t)-r(0))\right)$.

As $\dfrac{di}{dt}=(R_0s-1)ci$, if the basic reproduction number satisfies 
$R_0>\dfrac{1}{s(0)}$, there will be an information outbreak with an increasing
number of sensors with the datum. In other words, $R_0$ determines whether
or not the information will spread through the network.

All these facts are summarized in the proposition below.
\begin{proposition}
\label{SIRprop1}
Consider a sensor network that aims to monitor a given area, and that has to
spread an alert or an information to a sink, whose presence is sporadic.
Suppose that an attacker tries to remove the datum in sensors' memory, and that:
\begin{enumerate}
  \item all sensor activities are negligible, in terms of energy,
  \item when a sensor has the datum, it spreads the
information to its neighbors with a probability $b$, until being attacked.
\end{enumerate} 
Denote by $T_t$ the typical time between transmissions, $T_e$ the typical time
an informed sensor loses its information due to the attacker, and by $s(0)$ the initial 
fraction of susceptible sensors.
So the information will spread through the network if and only if $T_t<s(0) T_e$.
\end{proposition}

In other words, this proposition states that if the reproduction ratio is
greater than one, then an ``epidemic'' occurs since the prevalence (the infected
ratio) increases to a peak and then decreases to zero. Otherwise there is
no epidemic since the prevalence decreases to zero.

\begin{figure}[ht]
\centering
\includegraphics[scale=0.6]{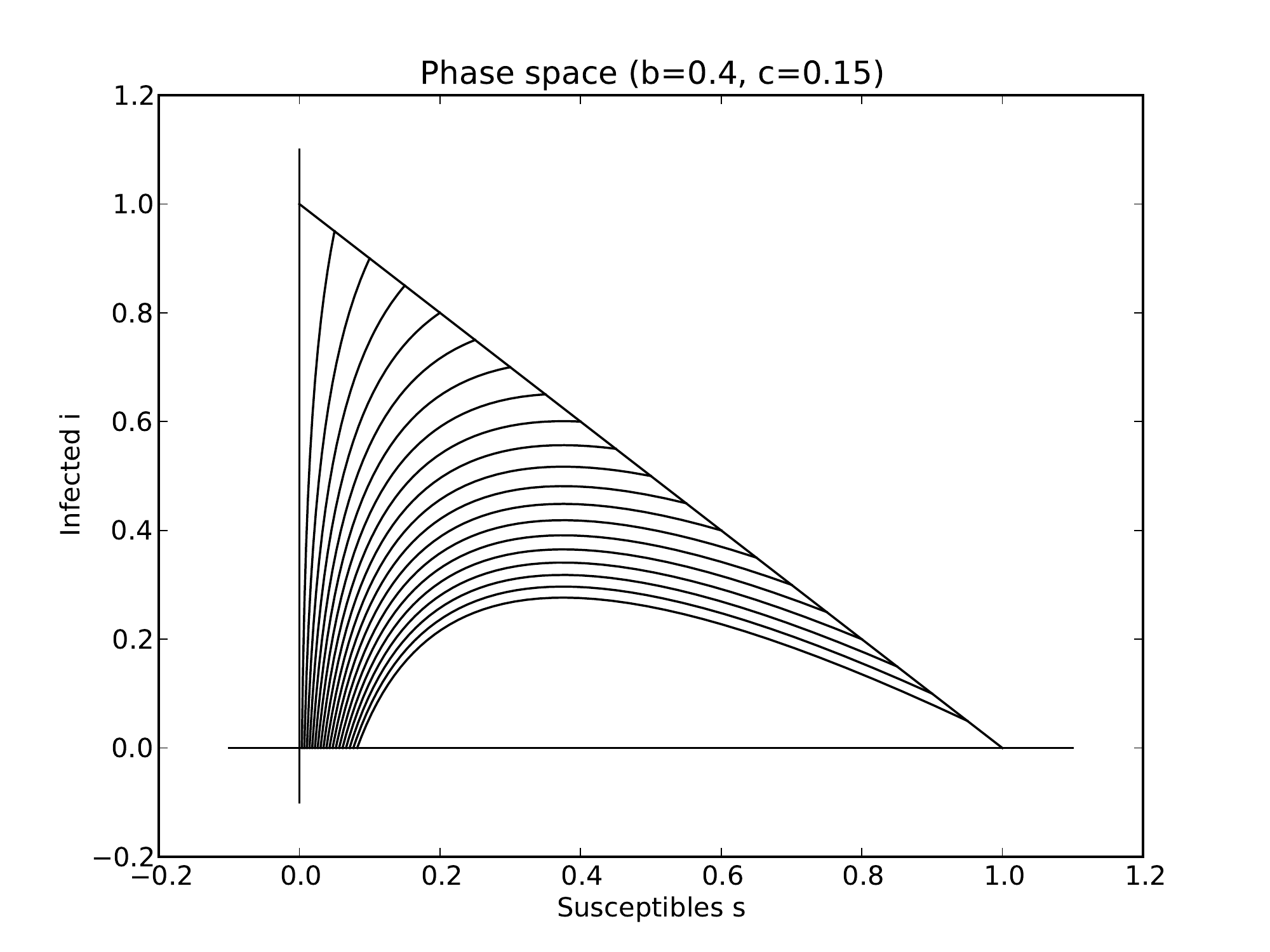}
\caption{Phase space $(s,i)$ with $b=0.4, c=0.15$ (SIR model).}
\label{img1}
\end{figure}

It is possible to be more precise in the formulation of Proposition~\ref{SIRprop1},
following an approach similar to~\cite{Hethcote00}.

\begin{proposition}
The fraction $s(t)$ of sensors susceptible to receive the information 
is a decreasing function. The limiting value
$s(\infty)$ is the unique root in $(0,\dfrac{T_e}{T_t})$ of the equation 
$$1-r(0)-s(\infty) + \dfrac{T_e}{T_t} ~ln\left(\dfrac{s(\infty)}{s(0)}\right).$$
Additionally,
\begin{itemize}
\item if $T_t\geqslant s(0) T_e$, then the fractional number $i(t)$ of sensors having the datum 
decreases to zero as $t \rightarrow \infty$,
\item else $i(t)$ first increases up to a maximum value equal to \linebreak $1-r(0)-\dfrac{T_e}{T_t}\left(1+ln\left(\dfrac{s(0)T_t}{T_e}\right)\right)$ and then decreases to
zero as $t\rightarrow \infty$, where $ln$ stands for the natural logarithm.
\end{itemize}
\end{proposition}

\begin{proof}
The triangle $T = \left\{(s,i) \mid s \geqslant 0, i \geqslant 0, s+i \leqslant 1\right\}$ is positively invariant, since from the SIR equations, it holds:
$s=0 \Rightarrow s'=0$, $i=0 \Rightarrow i'=0$, and 
$s+i=1 \Rightarrow (s+i)'= -ci \leqslant 0$.
Furthermore, points on the $s$ axis where $i=0$ are equilibrium ones, unstable
for $s>1/R_0$ and stable otherwise.
$s$ is decreasing and positive due to this invariance and 
because $\dfrac{ds}{dt} = - b i s$, so an unique limit $s(\infty)$ exists.
Similarly, $r'(t)=ci \geqslant 0$ and $r\leqslant 1$ then $r(\infty)$
exists. As $s+i+r=1$, $i(\infty)$ exists too. To prove that this limit is null,
we only remark that if $i(\infty)>0$, then $r(\infty)=\infty$ (because
$r'>\dfrac{c i(\infty)}{2}$ for sufficiently large $t$), which is impossible,
as $r \leqslant 1$. Finally, the equations of the proposition
are derived from $\dfrac{ds}{di} = \dfrac{c}{bs}-1$.
\end{proof}

The phase space of the solutions of the SIR system with given parameters is provided 
in Figure~\ref{img1}
while the evolution of $s$ and $i$ is depicted in Figure~\ref{img2}.
\begin{figure}[ht]
\centering
\includegraphics[scale=0.6]{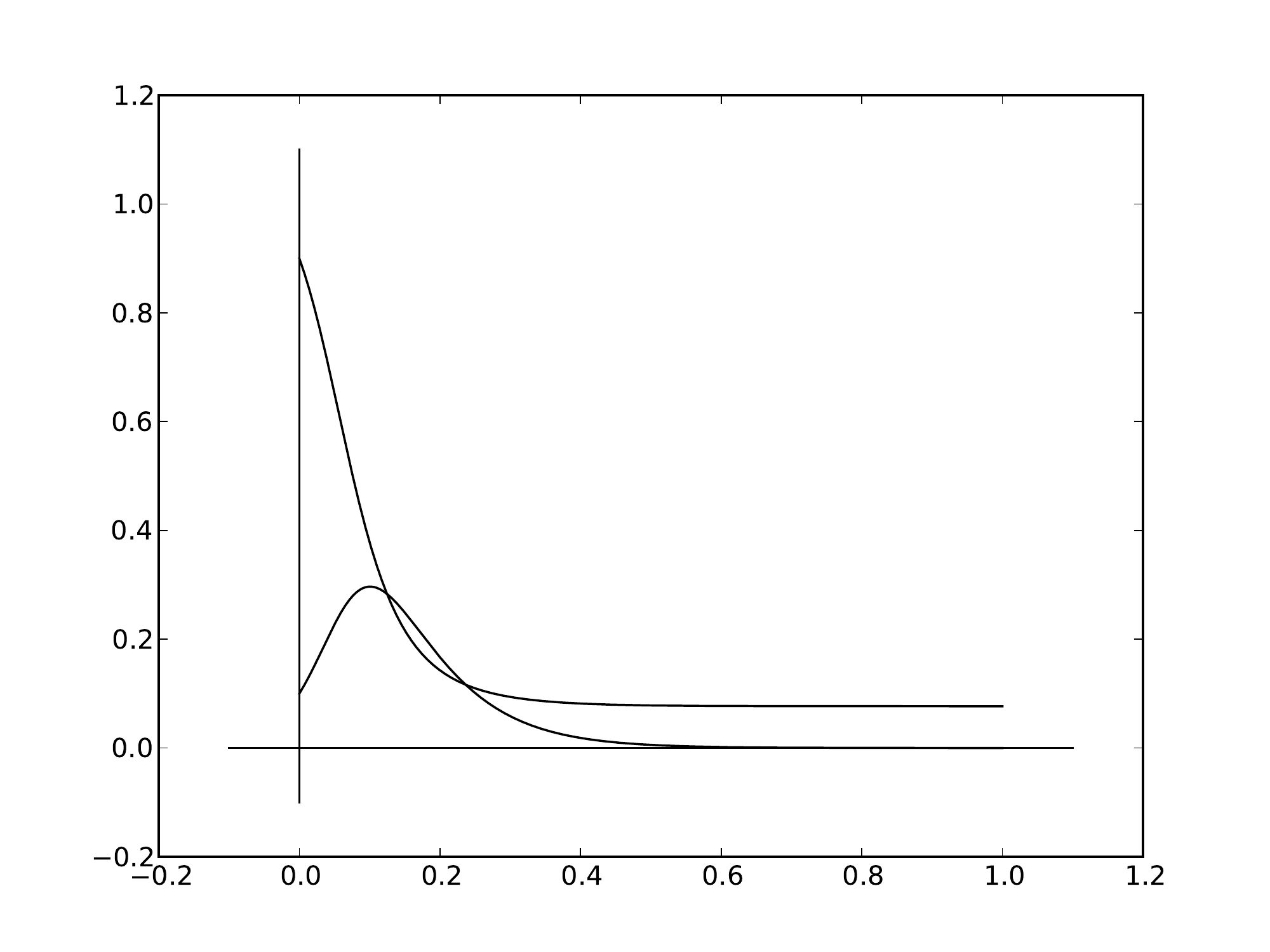}
\caption{Evolution of the fractions $s$ and $i$ of susceptible and having the datum 
sensors with $b=0.4, c=0.15, s(0)=0.9,$ and $i(0)=0.1$  (SIR model).}
\label{img2}
\end{figure}
The results presented in this section hold for a
 transition rate between susceptible and informed sensors
 having the form $F=ai$, which thus represents the force of information.
Nonlinear forces of information, or infection, can be investigated too, to model
more realistically the information survivability.

\subsection{Another understandings for the Recovered compartment}

In the previous section, the $R$ compartment was constituted by 
sensors  that have been compromised by the attacker, which will
be referred in what follows as situation 1.
It is possible to attribute at least two other understandings to
this compartment, for an unattended wireless sensor network 
whose lifetime is dependent on energy consumption and in absence of attacks.

This compartment can be constituted by dead sensors, when considering 
that the sole action on the energy is the information transmission, and
that the unique way to death for a sensor is to have too much transmitted the datum. In other
words, in this Situation 2, sensors send information messages to their neighbors until emptying
totally their batteries. The sink will receive the information when it
will interrogate the network at time $t$ if $I(t) \neq 0$.

A third situation can be considered without any changes in formalization,
except redefining the meaning of the $R$ compartment. Indeed, it can be
interesting to consider that a sensor is first susceptible to receive an information message
for a while, then in a second time it has and transmit the information, before 
finally entering into the third age of its life, the recovered state in
which it will lose its ability to transmit the information. Materials of the
previous section tackles too this scenario, when considering the network
lifetime sufficiently large compared to information spreading, in order 
to neglect sensors' death due to energy consumption. The question here
is to determine the quantity of informed sensors on large timescales.

Let us now explain how to extend such a compartmental study for data
survivability in wireless sensor networks to well-known SIS models.

\subsection{A few words about SIS models}

Other compartmental divisions of the set of sensors can be investigating,
leading for instance to a SIS epidemic model~\cite{Kermack27}. This latter 
assumes only two compartments named Susceptible (S) and
Infected (I). Transitions between these compartments are
represented in Figure~\ref{SISmodel}. An individual that is susceptible
to a disease becomes infected with a certain probability $a$, 
while an infected individual immediately becomes 
susceptible once (and if) it is cured of an infection (which
happens with probability $b$). Note that a healthy individual
can contract a disease only if it is in contact with a sick one.
Thus, the evolution of this system is completely described
by the following two differential equations (total sensor 
population: $P=S+I=S_0+I_0$, which is a constant).

$$\left\{
\begin{array}{ll}
\dfrac{dI}{dt} = a S I - b I & I(0) = I_0\\\\
\dfrac{dS}{dt} = b I - a S I& S(0) = S_0\\
\end{array}
\right.$$

The SIS model may be treated the same as the SIR model, which 
has been detailed in this section. For the sake of concision, 
and as this study does not raise any complication, this model
will be left as an exercise, while energy consumption will 
now be investigated in the next section.

\begin{figure}[ht]
\centering
\includegraphics{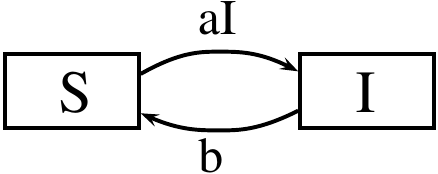}
\caption{SIS model}
\label{SISmodel}
\end{figure}

\section{Considering Energy Consumption for Data Survivability in UWSNs}
\label{USIR}

In a large amount of situations, energy consumption and the
death of sensors cannot be neglected, this is why a ``natural'' death rate for all
compartments is introduced in this section. Such an approach generalizes
the models presented previously.

\subsection{A SIR model with natural death rate}

\begin{figure}[ht]
\centering
\subfigure[Situation 2]{\label{SIR2sit2}
\includegraphics[scale=0.65]{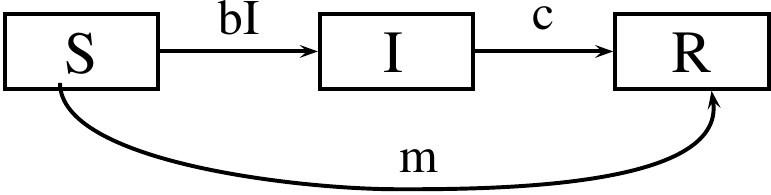}}\hspace{1cm}\subfigure[Situations 1 and 3]{\label{SIR2sit13}
\includegraphics[scale=0.65]{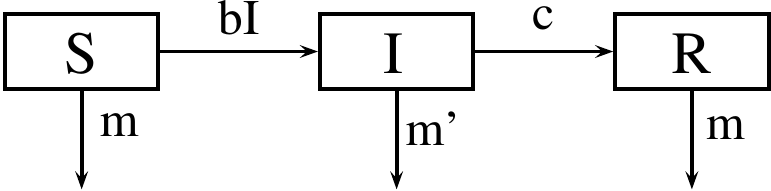}} 
\caption{SIR models with natural death rate}
\label{SIR2modeldeath}
\end{figure}

The previous section considers that all sensor activities are negligible, in terms of 
energy, except the transmission of information in situations 2 and 3, which is reasonable in a
first approximation. It is however possible to refine 
the SIR model in these two last situations, in order to 
consider that sensors' energy decreases too in absence of information transmission.

\begin{figure}[ht]
\centering
\subfigure[Situation 2]{\includegraphics[scale=0.35]{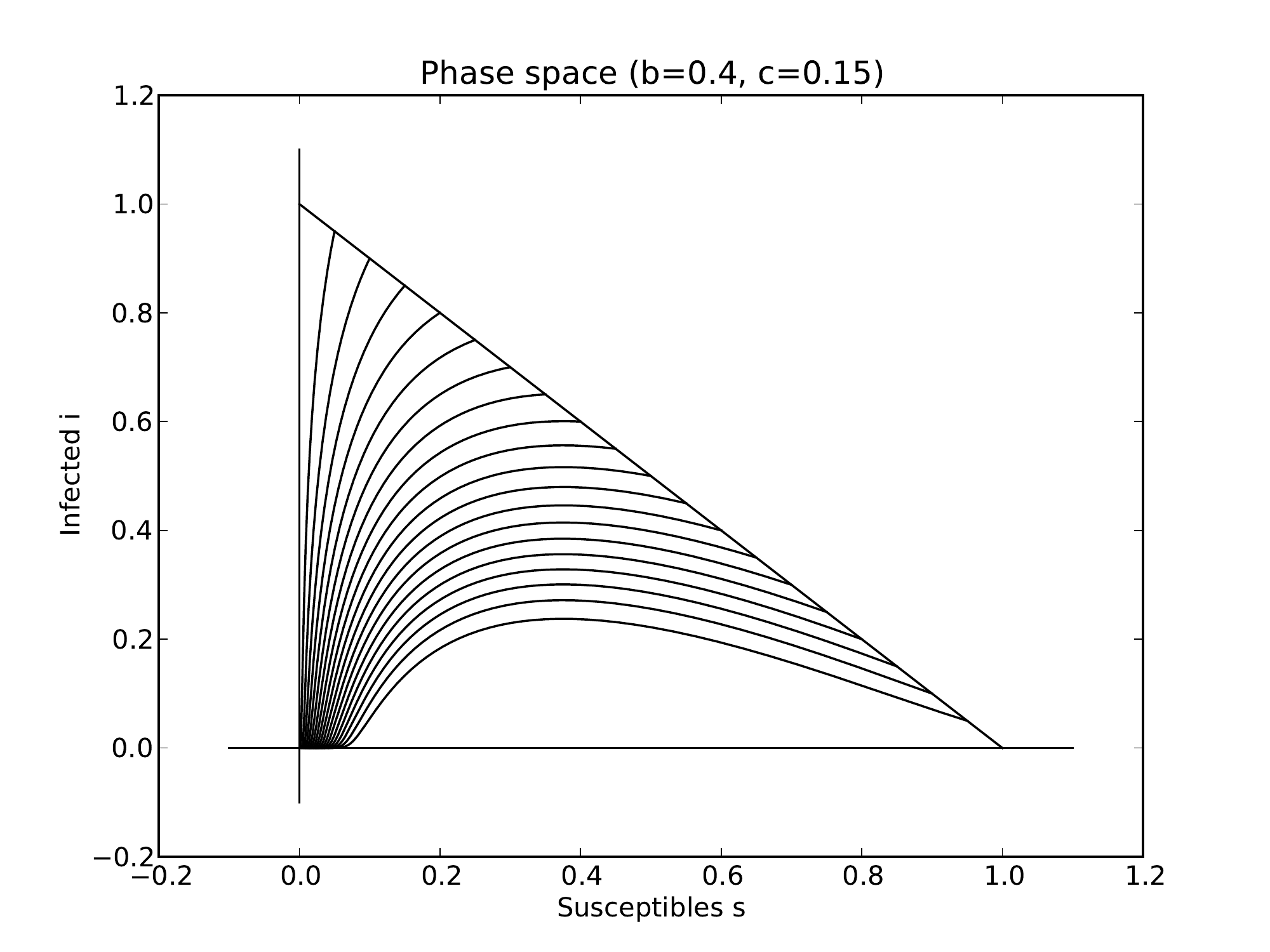}}\subfigure[Situations 1 and 3]{\includegraphics[scale=0.35]{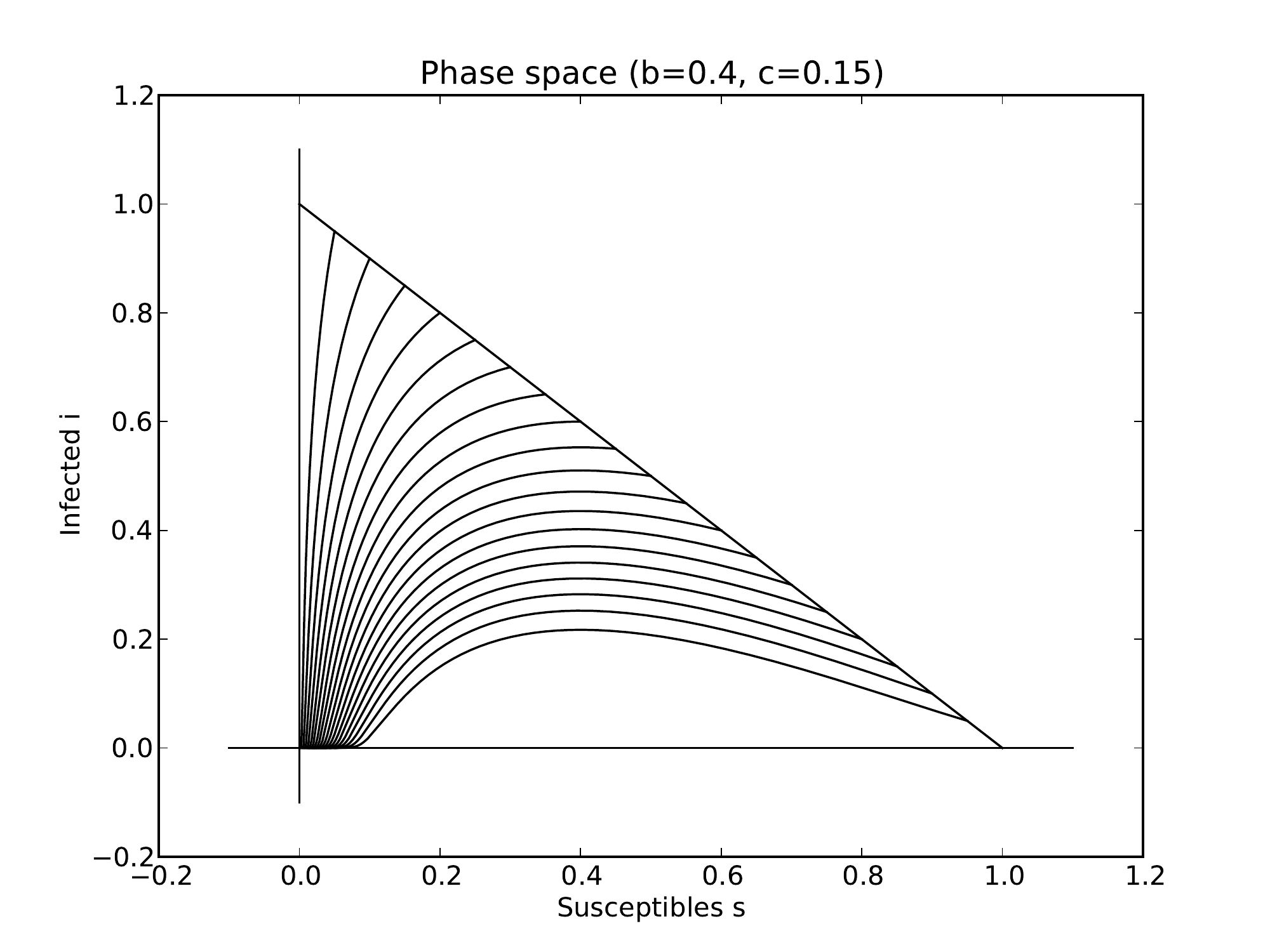}}
\caption{Phase space $(s,i)$ with $b=0.4, c=0.15, m=0.01$, SIR model with natural death rate in Situation 3.}
\label{img1}
\end{figure}

In Situation 2, the $R$ compartment of the SIR model is constituted by 
dead sensors. This compartment is populated by susceptible nodes that have
naturally died (death rate $m$) without having received the datum and by
sensors of the $I$ compartment which die at another rate $c$ supposed
to be greater than $m$, as they have to transfer the datum, an energy-consuming
task. This situation is depicted in Figure~\ref{SIR2sit2}.

In the two other situations investigated in this research work, the $R$ 
compartment is constituted by living sensors that do not transmit the
datum anymore, either because they have been corrupted and thus have 
lost it (first situation), or because their batteries is preserved (third one).
This new situation is closed to the SIR model of Figure~\ref{SIRmodel}, except
that a the new network is characterized by a death rate for each sensors
compartment (see Figure~\ref{SIR2sit13}). 
Notice that the death rate $m'$ of the $I$ compartment is \emph{a priori} different
from the one of $S$ and $R$ compartments, as it is reasonable to suppose that
the datum transmission implies more energy consumption. However, setting $m'=m$
is possible too.

The SIR model of Equation~\eqref{modelSir1} can be adapted as follows for Situation 2:
\begin{equation}
\label{modelSir2}
\left\{
\begin{array}{l}
\dfrac{ds}{dt} = - b i s -m s\\\\
\dfrac{di}{dt} = b i s - c i\\\\
\dfrac{dr}{dt} = c i+ms,\\
\end{array}
\right.
\end{equation}
while it has the following form in Situations 1 and 3:
\begin{equation}
\label{eqSit3}
\left\{
\begin{array}{l}
\dfrac{ds}{dt} = - b i s -m s\\\\
\dfrac{di}{dt} = b i s - c i - m' i\\\\
\dfrac{dr}{dt} = c i - mr.\\
\end{array}
\right.
\end{equation}

Let us now investigate the long-term behavior of these models.
Regarding Situation 2, it is natural to think that, for large timescales, all sensors will 
take place in the third $R$ compartment of died sensors, as all the batteries are
continually emptied (either due to natural consumption or because of the information transmission).
This can be easily proven by considering that in an equilibrium point $(s^*,i^*,r^*=1-s^*-i^*)$, we have $\dfrac{ds}{dt} = \dfrac{di}{dt} = \dfrac{dr}{dt} = 0$, and so
$$
\left\{
\begin{array}{l}
(bi^*+m)s^* = 0\\
(bs^*-c)i^*=0\\
ci^*+ms^*=0.\\
\end{array}
\right.
$$
As $c>0, m>0, i^*\geqslant 0$, and $s^*\geqslant 0$, we can conclude from the third
equation above that $s^*=i^*=0$, and so $r^*=1$.
The Jacobian is equal to 
$$J(s,i,r)=\left( 
\begin{array}{ccc}
-bi-m & -bs & 0\\
0 & bs-c & 0 \\
m & c & 0
\end{array}
\right)$$
and its characteristic polynomial in $(0,0,1)$ is $\lambda (\lambda + c) (\lambda +m)$. The 
eigenvalues being negative, the equilibrium $(0,0,1)$ is attractive.
These results are summarized in the following proposition.

\begin{figure}[ht]
\centering
\includegraphics[scale=0.6]{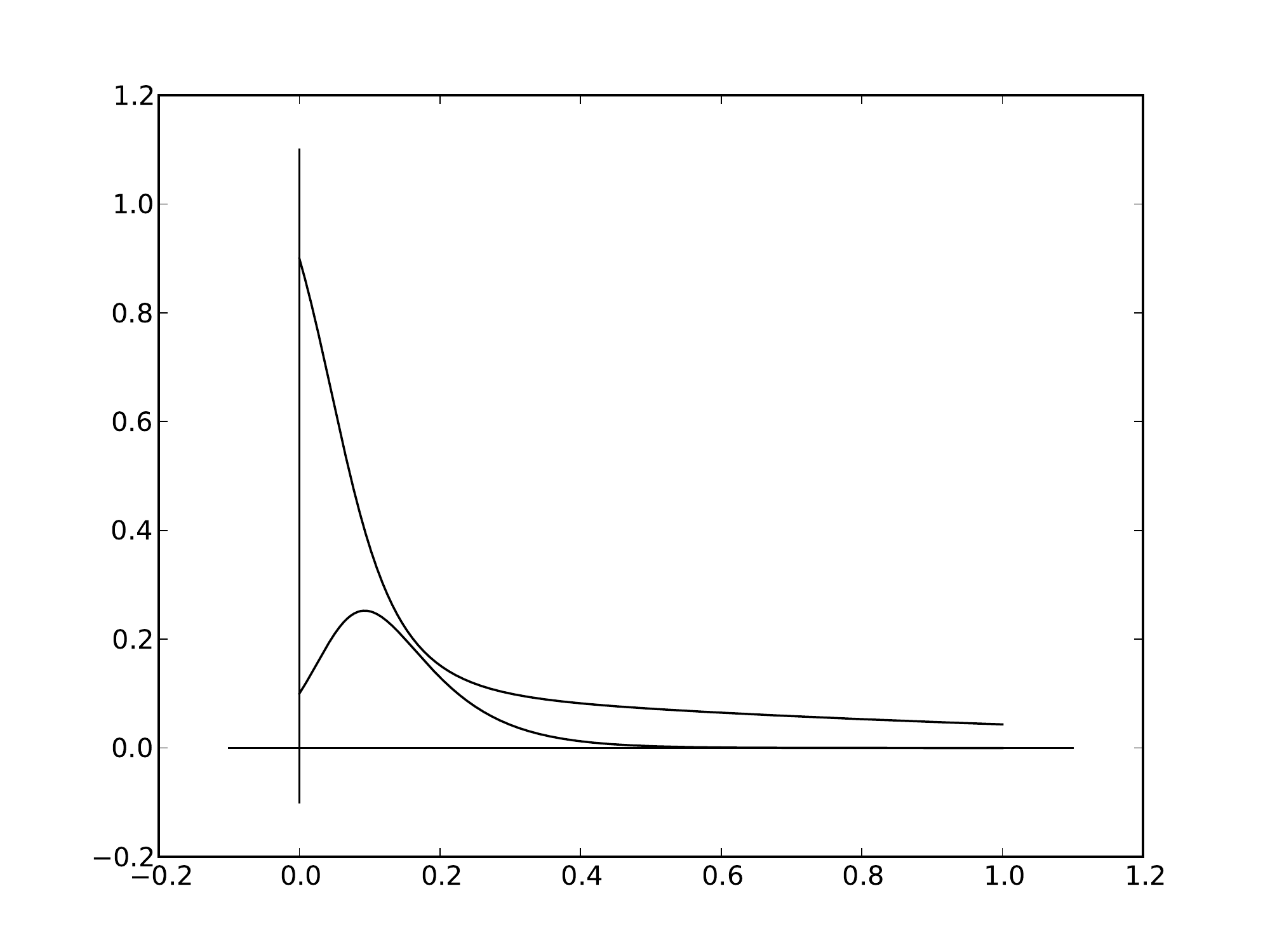}
\caption{Evolution of the fractions $s$ and $i$ of susceptible and having the datum 
sensors with $b=0.4, c=0.15, m=0.01, s(0)=0.9,$ and $i(0)=0.1$, SIR model with natural death rate in Situations 1 and 3.}
\label{img2}
\end{figure}

\begin{proposition}
Consider an unattended wireless sensor network divided in three sets of sensors,
the first category $S$ being susceptible to receive a given datum, the second one 
$I$ having and transmitting this latter, and the third one $R$ being constituted
by dead sensors. 

Suppose that the death rate is $m$ for $S$ compartment and $c$ for $I$'s one, and
that the transmission rate is $bI$ between $S$ and $I$. In that situation, for all
initial condition and for all positive parameters $b, c,$ and $m$, the system is
convergent to the equilibrium point $(0,0,1)$.

In particular, in that situation, the datum cannot survive a long time in the UWSN.
\end{proposition}

Equation~\ref{eqSit3} can be resolved similarly: from $bi^*s^*+ms^*$, 
we deduce that $s^*=0$ (as b>0, m>0, and $i^* \geqslant 0$). So
$bi^*s^*-ci^*-m'i^*$ implies that $i^*=0$ too. Finally, from 
the third line, we conclude that $r^*=0$.
Eigenvalues of the characteristic polynomial of the Jacobian
in $(0,0,0)$ are $-m$ and $-c-m'$, which are negative. So this
equilibrium point is attractive too, and a similar proposition 
than previously can be formulated, with the same conclusion, both
for Situations 1 and 3. Phase spaces for the three situations
are provided in Figure~\ref{img1} while Fig.~\ref{img2}
depicts the evolution of the fractions $s$ and $i$ in Situations
1 and 3.

To put it in a nutshell, to achieve data survivability in 
UWSNs, the birth of awaken sensors must be considered, which 
is the subject of the next subsection.

\subsection{A scheduling process in data survivability}

\subsubsection{A first natural approach}

A first idea to realize a more realistic model of an unattended 
wireless sensor network is to establish a scheduling process
of the sensor nodes, in order to enhance data survivability for a long period of time.
Considering the SIR model, such a process
leads to the division of each compartment in two parts, corresponding
respectively to awaken and to sleeping sensors, as depicted in
Figure~\ref{momodelSir}.

\begin{figure}[ht]
\centering
\includegraphics{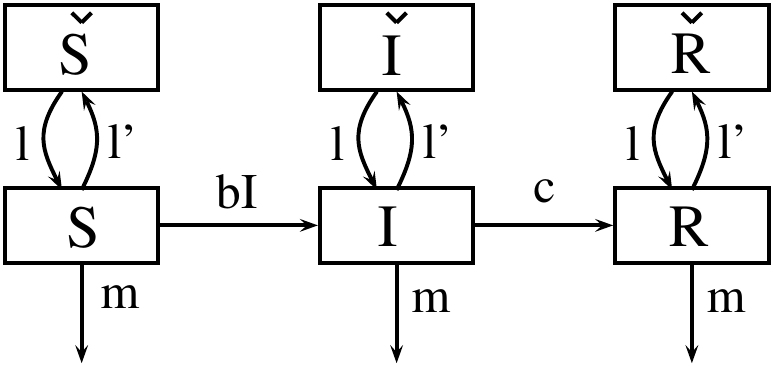}
\caption{SIR model with natural death rate and sleeping nodes}
\label{momodelSir}
\end{figure}

Such a model can be reformulated as follows:
\begin{equation}
\label{modelSir3}
\left\{
\begin{array}{ll}
\dfrac{ds}{dt} = l\check{s}-l's - b i s -m s & 
\dfrac{d\check{s}}{dt} = -l\check{s}+l's \\\\
\dfrac{di}{dt} = l\check{i}-l'i+b i s - c i - m i & 
\dfrac{d\check{i}}{dt} = -l\check{i}+l'i \\\\
\dfrac{dr}{dt} = l\check{r}-l'r+c i - m r & 
\dfrac{d\check{r}}{dt} = -l\check{r}+l'r. \\
\end{array}
\right.
\end{equation}

The equilibrium point $(s^*,\check{s}^*,i^*,\check{i}^*,r^*,\check{r}^*)$ is searched once again, it satisfies:
\begin{equation}
\left\{
\begin{array}{ll}
l\check{s}^*-l's^* - b i^* s^* -m s^* = 0 & 
-l\check{s}^*+l's^* = 0\\\\
l\check{i}^*-l'i^*+b i^* s^* - c i^* - m i^* = 0 & 
-l\check{i}^*+l'i^* = 0\\\\
l\check{r}^*-l'r^*+c i^* - m r^* = 0 & 
 -l\check{r}^*+l'r^* = 0. \\
\end{array}
\right.
\end{equation}
Obviously, $l\check{s}^*=l's^*$, $l\check{i}^*=l'i^*$, and $l\check{r}^*=l'r^*$, and so:
\begin{equation}
\left\{
\begin{array}{l}
(b i^* + m) s^* = 0\\\\
(b s^* - c - m) i^* = 0\\\\
c i^* - m r^* = 0.\\
\end{array}
\right.
\end{equation}
If $s^* \neq 0$, then $b i^* + m = 0$, which is impossible 
if it is reasonably supposed that each rate is $>0$. So $s^*=0$,
which implies that $i^*=0$, and so $r^*=0=\check{r}^*=\check{i}^*=\check{s}^*$.

To sum up, in the unique stable equilibrium point, the number of informed
sensors is null, and we face a data loss. This problem is solved in the
next section, by considering that nodes never go to sleep.

\subsubsection{Achieving data survivability using birth and death rates}

\begin{figure}[ht]
\centering
\includegraphics{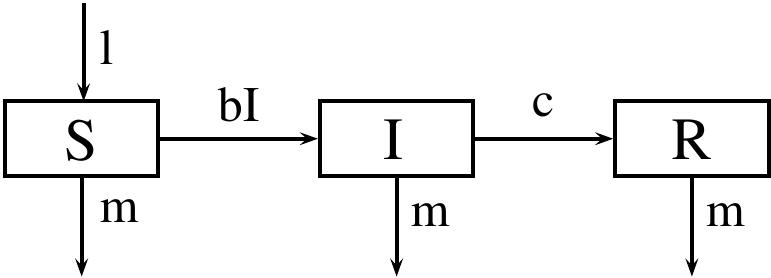}
\caption{SIR model with natural birth and death rates}
\label{SIR3modelbirthdeath}
\end{figure}

Consider now a new approach proposed to solve
the loss of information in the former scheduling process.
In this second approach for scheduling, sensors can only be
awaken (we never order them to sleep). It is supposed that
a sufficiently large number of sensors are available, and the question
is to determine if it is possible to determine the lowest 
birth rate to achieve data survivability for a long period of time 
, even in presence of an adversary.

To do so, it is supposed that, at the initial stage, only a small part of 
the sensors nodes is awakened.
New sensors are then awakened periodically during the 
network's service
at a rate $l$, repopulating by doing so the $S$
compartment (they never go to sleep). 
Along with this birth rate, a natural death rate $m$ is considered 
for each of the
three kind of sensors, while the $R$ compartment is for corrupted 
sensors in the original situation 1, 
as depicted in Figure~\ref{SIR3modelbirthdeath}. 
Remark that such a model is compatible with living and awaken
nodes that have stopped to transfer the information in Situation 3.

To model such a scenario requires to rewrite the first line of 
Equation~\eqref{modelSir2}, leading to the following system:
\begin{equation}
\label{modelSir3}
\left\{
\begin{array}{l}
\dfrac{ds}{dt} = l - b i s -m s\\\\
\dfrac{di}{dt} = b i s - c i - m i\\\\
\dfrac{dr}{dt} = c i - m r.\\
\end{array}
\right.
\end{equation}

%
%
This updated system is the usual SIR model with vital dynamics, in which
we have not supposed the birth and death rates equal.
It is possible to show that the problem is well formulated, as
the triangle 
$T = \left\{(s,i) \mid s \geqslant 0, i \geqslant 0, s+i \leqslant 1\right\}$
 still remains positively invariant.

A study of this system supposes to consider the 
Poincar\'{e}-Bendixon theorem in phase space and the use of
Lyapunov functions~\cite{Hethcote00}. It can however be 
understood by considering what will happen to the information in a long run: 
will it die out or will it establish itself in the network like an
endemic situation in epidemiological models?
The long-term behavior of the solutions, 
which depends largely on the equilibrium points
that are time-independent solutions of the system,
must be investigated to answer this question.
Since these solutions do not depend on time, we have $s'(t)=i'(t)=r'(t)=0$,
which leads to the system:
$$\left\{
\begin{array}{l}
0 = l - b i s -m s\\\\
0 = b i s - (c+m) i\\\\
0 = c i -m r.\\
\end{array}
\right.$$
$r=\dfrac{c}{m}i$ from the last equation, and either $i=0$ or $s=\dfrac{c+m}{b}$
from the second one. 
On the one hand, if $i=0$, then $r=0$, and $s=\dfrac{l}{m}$ 
from the first equation. This leads to the equilibrium solution $$\left(\dfrac{l}{m},0,0\right).$$
As the number of sensors having the datum is 0 in this point, it means that
if a solution of the system approaches this equilibrium, the fraction $i$ will
approach 0, and the datum tends to disappear from the network: an \emph{information-free equilibrium}. 
Remark that the existence of this equilibrium is independent of the parameters of the system: it always exists.

On the other hand, if $i\neq 0$, then $s=\dfrac{c+m}{b} \neq 0$ from the second equation, and $\dfrac{l}{s}=bi + m$ according to the first equation. Substituting
$s$ and solving for $i$, we find $$i=\dfrac{bl-m(c+m)}{b(c+m)}=\dfrac{R_0l-m}{b},$$
with $R_0=\dfrac{bl}{m(c+m)}$, which is a positive number iff $R_0>1$.

$R_0$ is the reproduction number of the information, which
tells us how many secondary informed sensors
will one informed sensor produces in an entirely susceptible network, as:
\begin{itemize}
\item a network which consists of only susceptible nodes in a long run has
$\dfrac{l}{m}$ sensors;
\item $c+m$ is the rate at which sensors leave the $I$ compartment. In
other words, the average time spent as an informed sensor is $\dfrac{1}{c+m}$
time units.
\item The number of data transmissions per unit of time is given by the 
incidence rate $bIS$. If
there is only one informed sensor ($I = 1$) and every other sensor is 
susceptible $(S = \dfrac{l}{m}$) then the number of
transmissions by one ``infected'' node per unit of time is $\dfrac{bl}{m}$.
\end{itemize}
So the number of data transmissions that one informed sensor can achieve during the
entire time it is not attacked if all the reminded sensors are susceptible, is
$\dfrac{bl}{m(c+m)}$, that is, $R_0$.

So if $R_0>1$, the number of sensors having the datum is strictly positive in this equilibrium
solution: if some other solutions of the system approach this equilibrium as time
goes large, the number of sensors having the datum will remain strictly positive, and
the information remains in the network and becomes endemic.

These statements are summarized in the following proposition.
\begin{proposition}
\label{laprop}
If either $R_0\leqslant 1$ or $s(0)=0$, then any solution $(s(t),i(t))$
is convergent to the equilibrium without information $(1,0)$. 

If $R_0>1$ then there are two equilibria: the non attractive information-free 
equilibrium and the endemic equilibrium. This latter is attractive
so that solutions of the ODE system approach it as time goes to infinity: the
information remains endemic in the UWSN.
\end{proposition}

\begin{figure}[ht]
\centering
\includegraphics[scale=0.5]{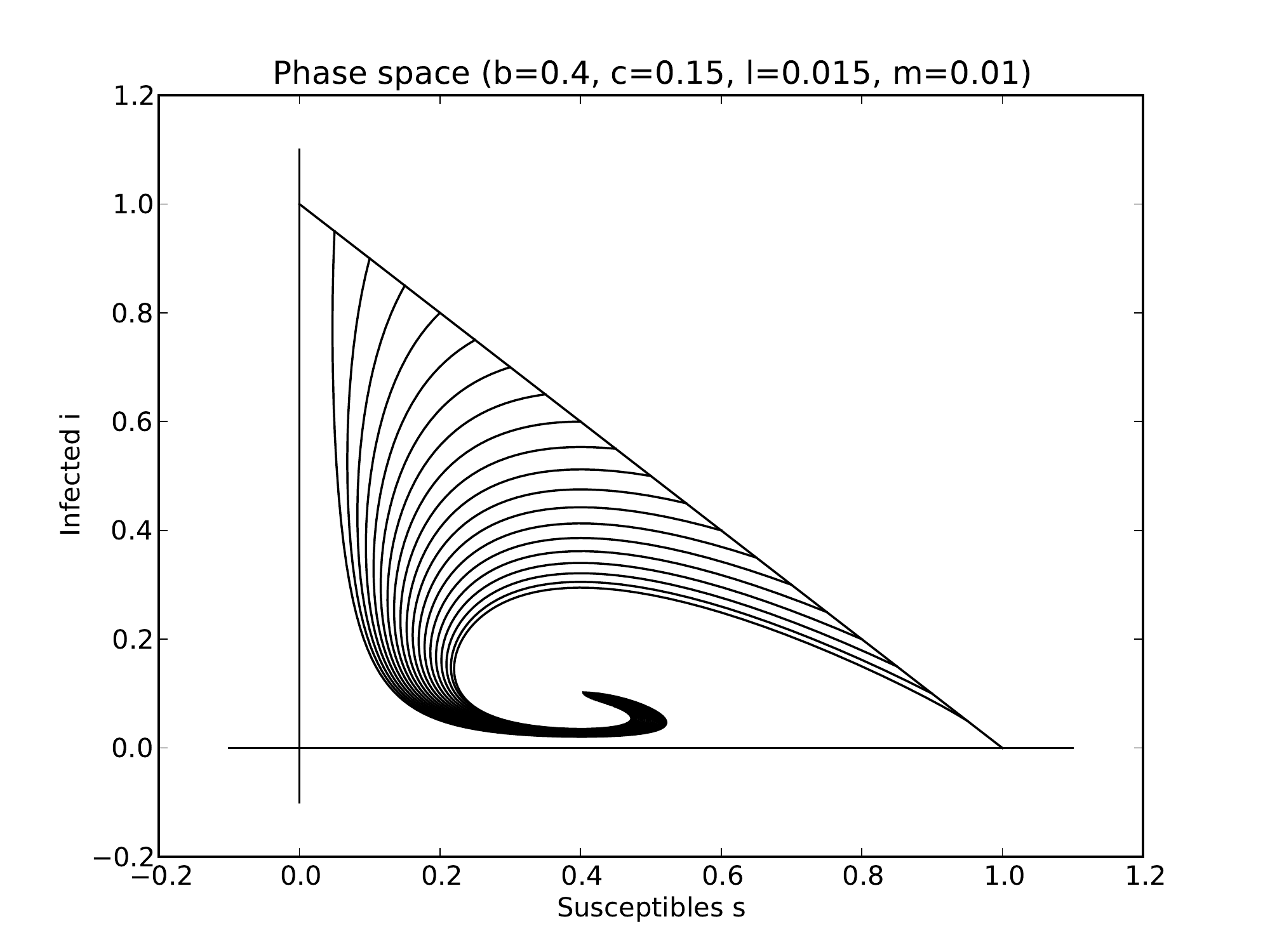}
\caption{Evolution of the fractions $s$ and $i$ of susceptible and having the datum 
sensors, SIR model with natural birth and death rates ($R_0=3.75$).}
\label{img2}
\end{figure}

The attacker desire is to have $R_0<1$ to tend to an information-free equilibrium,
whereas $R_0$ must be greater than 1 for the sink to face such attack. If the
attacker has the opportunity to observe the network running a certain duration,
then he or she can infer the values of parameters $b, c, m$, and $l$. 
Let $N$ be the number of data transmissions by one
informed node per time unit, that is, $N=\dfrac{bl}{m}$. If the attacker
is able to detect and infect the informed nodes in a time $\dfrac{1}{c+m}$ lower than $\dfrac{1}{N}$,
then he or she is sure that $R_0<1$: the data will not survive in the network.
The sink interest, for its part, is to have $\dfrac{bl}{m}$ large and $\dfrac{1}{c+m}$ low,
which can be achieved in the following manner:
\begin{itemize}
\item increasing the birth rate $b$,
\item increasing the lifetime of sensors to reduce $m$,
\item increasing the data transmission rate $b$, but $m$ increases when $b$ increases,
\item if possible, reducing $c$ by considering countermeasures against data removal.
\end{itemize}

\begin{figure}[ht]
\centering
\includegraphics{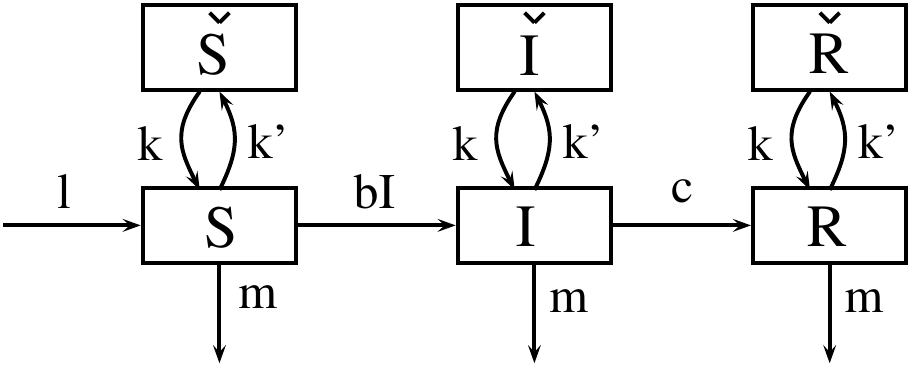}
\caption{Global SIR model with natural birth and death rates, and sleeping nodes}
\label{momodelSir2}
\end{figure}

Remark finally that this study is compatible with the situation depicted in Figure~\ref{momodelSir2},
in which awaken sensor nodes are allowed to go to sleep. Indeed this situation, which
has not been detailed in this section to avoid making the text more cumbersome,
introduces three new compartments $\check{S}, \check{I},$ and $\check{R}$ as in the
previous section. However, as we focused on the future of the information in a long run,
we only have to consider equilibrium points that are time-independent solutions of the 
system. As shown in the previous section, we obtain 
$\dfrac{d\check{s}}{dt} = -k\check{s}+k's = 0$, 
$\dfrac{d\check{i}}{dt} = -k\check{i}+k'i = 0$,
and $\dfrac{d\check{r}}{dt} = -k\check{r}+k'r$. Consequently,
compartments $\check{S}, \check{I},$ and $\check{R}$ disappear 
in the final global system corresponding to figure~\ref{momodelSir2}:
\begin{equation}
\label{modelSir3}
\left\{
\begin{array}{ll}
\dfrac{ds}{dt} = l +k\check{s}-k's- b i s -m s & 
\dfrac{d\check{s}}{dt} = -k\check{s}+k's \\\\
\dfrac{di}{dt} =  k\check{i}-k'i+ b i s - c i - m i &
\dfrac{d\check{i}}{dt} = -k\check{i}+k'i \\\\
\dfrac{dr}{dt} =  +k\check{r}-k'r+c i - m r & 
\dfrac{d\check{r}}{dt} = -k\check{r}+k'r, \\
\end{array}
\right.
\end{equation}
and exactly the same Proposition~\ref{laprop} is obtained.

\section{The proposed algorithm}
\label{sec:algo}

In this section, a fully distributed algorithm which supports/covers different epidemic models is presented and theoretically analyzed. Our algorithm seeks to ensure {\it data survivability} by maintaining a necessary 
set of safe working nodes and replacing/locking attacked ones when needed. 
 
In the following, we first focus on the legitimate state formulation and next, we present the
algorithm which consists in only three rules and give the correctness proofs. 

\subsection{Problem formalization}
Let $G = (V;E)$ the graph modeling the sensor network, with $|V| =n$ and $|E| = m$. We assume sensor node identifiers to be unique. Recall that sensor node identifier is unique if and only if $i.Id \neq j.Id$ holds for each $i, j \in V (i \neq j )$. A sensor node can be in one of these four states: \textit{working}, \textit{probing}, \textit{sleeping} or locked.

\smallskip

 \noindent We say that a sensor node $i$ is independent if
\begin{center}
$i.state = working \wedge \left(\forall j \in N_i\right) \left(j.state = sleeping \vee probing \vee locked \right)$
\end{center}

\noindent and that $i$ is dominated if
\begin{center}
$\left(i.state = sleeping\vee probing \vee locked \right)\wedge \left(\exists j \in N_i\right) \left(j.state = working\right)
$
\end{center}

  The legitimate state (let~denote~it~$\cal L$) of the network is then expressed as follows:
$$\forall i \in V : i.state = working \Rightarrow i.compartment = S \vee I$$
In other words, each working node is either in $S$ or $I$. 

\medskip

The following notations are also given for the predicates of node $i$\\

\noindent 
- $A(i)$: attacked neighbor: {\small $ ~\exists~j \in N_i, j.compartment = R$}

\smallskip
\noindent 
- $W(i)$: working neighbor:  {\small$ ~\exists~j \in N_i, j.state = working$}

\smallskip
\noindent 
- $ W^{*}(i)$: working neighbor with lower Id:  {\small $\exists~j \in N_i, j.state = working \wedge i.Id > j.Id$}

\smallskip
\noindent 
- $ P^{*}(i)$: probing neighbor with lower Id: {\small
$\exists~j \in N_i, j.state = probing \wedge i.Id > j.Id$}

\subsection{The algorithm}
\medskip
\noindent The proposed algorithm uses the following three rules:

\noindent
{\bf $r_1$}: 
\begin{algorithmic} 
\IF {$i.state = probing \wedge W(i)$}
\IF {$j.compartment = I$}
\STATE {$i.compartment \leftarrow I $~{\footnotesize \tt (*the datum is transferred/replicated to/on $i$*)}}
\ENDIF
\STATE {$i.state  \leftarrow sleeping$}
\ENDIF
\end{algorithmic} 

\medskip
\noindent
{\bf $r_2$}: 
\begin{algorithmic} 
\IF {$i.state = probing \wedge \left(\neg W(i)\wedge \neg P^*(i) \vee A(i)\right)$}
\IF {$A(i)$}
\STATE {$j.state \leftarrow locked$~{\footnotesize \tt (*node $j$ remains locked until its healing/recovery*)}} 
\ENDIF
\STATE {$i.state \leftarrow working$}
\ENDIF
\end{algorithmic}

\medskip
\noindent
{\bf $r_3$}: 
\begin{algorithmic} 
\IF {$ i.state = working \wedge  W^*(i)  $}
\IF {$i.compartment = S \wedge j.compartment = I$}
\STATE {$i.compartment \leftarrow I $~{\footnotesize \tt (*the datum is transferred/replicated to/on $i$*)} }
\ENDIF
\STATE {$i.state \leftarrow sleeping$}
\ENDIF
\end{algorithmic}

\subsection{Correctness proofs}
 
 \smallskip
  \begin{lemma} If a node changes to the {\it working} state by $r_2$, then it remains in its state and will never execute a rule again until an eventual attack.

\end{lemma}

\begin{proof} Let $i$ be a sensor node that executes $r_2$. According to the preconditions of all rules, node $i$ can execute 
only rule $r_3$ in the next round. However, in order to do so, one of its neighbors would have to change 
into {\it working} state by $r_2$. This is impossible as long as node $i$ is in the {\it working} state. Thus, node $i$ will never execute a rule again. If node $i$ is attacked, it will be locked by $r_2$ and remains in its state until
its healing/recovery. After that, it will join the set of sleeping nodes.
\end{proof}

\begin{lemma} If a sensor node is enabled by rule $r_2$, then each one of its neighbors will execute at most one more rule until their next wakeup/probing, and this rule will be $r_1$.

\end{lemma}

\begin{proof} Let $i$ be a node that executes $r_2$.
When node $i$ changes to working state, all its neighbors are either in {\it sleeping} or {\it probing} or {\it locked} state. So we have three possible scenarios: i) neighbors in sleeping state: there is no conflict in this case.
ii) neighbors with probing state: those neighbors have a higher $Id$ than $i$. iii) locked neighbors will remain in their state until their healing/recovery before joining the set of sleeping nodes.

\end{proof}

\begin{lemma} Every sensor node is either independent or dominated or locked. 
\end{lemma}
\begin{proof} From the point of view of node $i$, we have three scenarios:\\
- if  node $i$ is in the $working$ state and is not  {\it independent}, then $i$ may execute rule~$r_3$.\\
- if  node $i$ is in the $sleeping \vee probing$ state and is not {\it dominated}, then  node $i$ may execute rule $r_2$.\\
- if  node $i$ is in the $locked$ state, then  node $i$ will remain in its state until its healing/recovery. 

\end{proof}

\begin{lemma}
When a node is not locked $\vee$ sleeping, it can make at most $2$ moves.
\end{lemma}

\begin{proof} By Lemma 1 and Lemma 2, each rule can be executed at most once by a node. Hence, the only case a node makes two moves is when it executes $r_3$ then $r_2$ with a \textit{working} state.

\end{proof}

\begin{theorem}
With respect to the legitimate state $\cal L$ of the network, the proposed algorithm converges within $2n$ moves.
\end{theorem}

\begin{proof} This follows from Lemma 1 to Lemma 4.

\end{proof}

\section{Simulations}

This section is dedicated to the evaluation of the SIR approach through experiments.
We will show, using both the mathematical modeling and a basic wireless sensor network
designed in Python, that taking place in the conditions of Proposition~\ref{laprop} 
is a guarantee to achieve information survivability in WSNs.

\subsection{Mathematics-based simulations}

In this first illustration written in Python language, the
initial number of susceptible sensors is set to 300
while 3 nodes initially receive the datum.
System~\ref{modelSir3} is then discretized and 4 experiments have
been conducted, leading twice to the situation $R_0<1$, and twice
to the opposite situation.

Figure~\ref{simul} shows the obtained result. We can see that the $I$ 
compartment is never empty when $R_0>1$, leading to a data
 survivability in this SIR model simulation. Conversely, 
 when $R_0<1$, the information is obviously lost.

\begin{figure}[ht]
\begin{center}
  \subfigure[$R_0 = 118.51$]{\includegraphics[scale=.35]{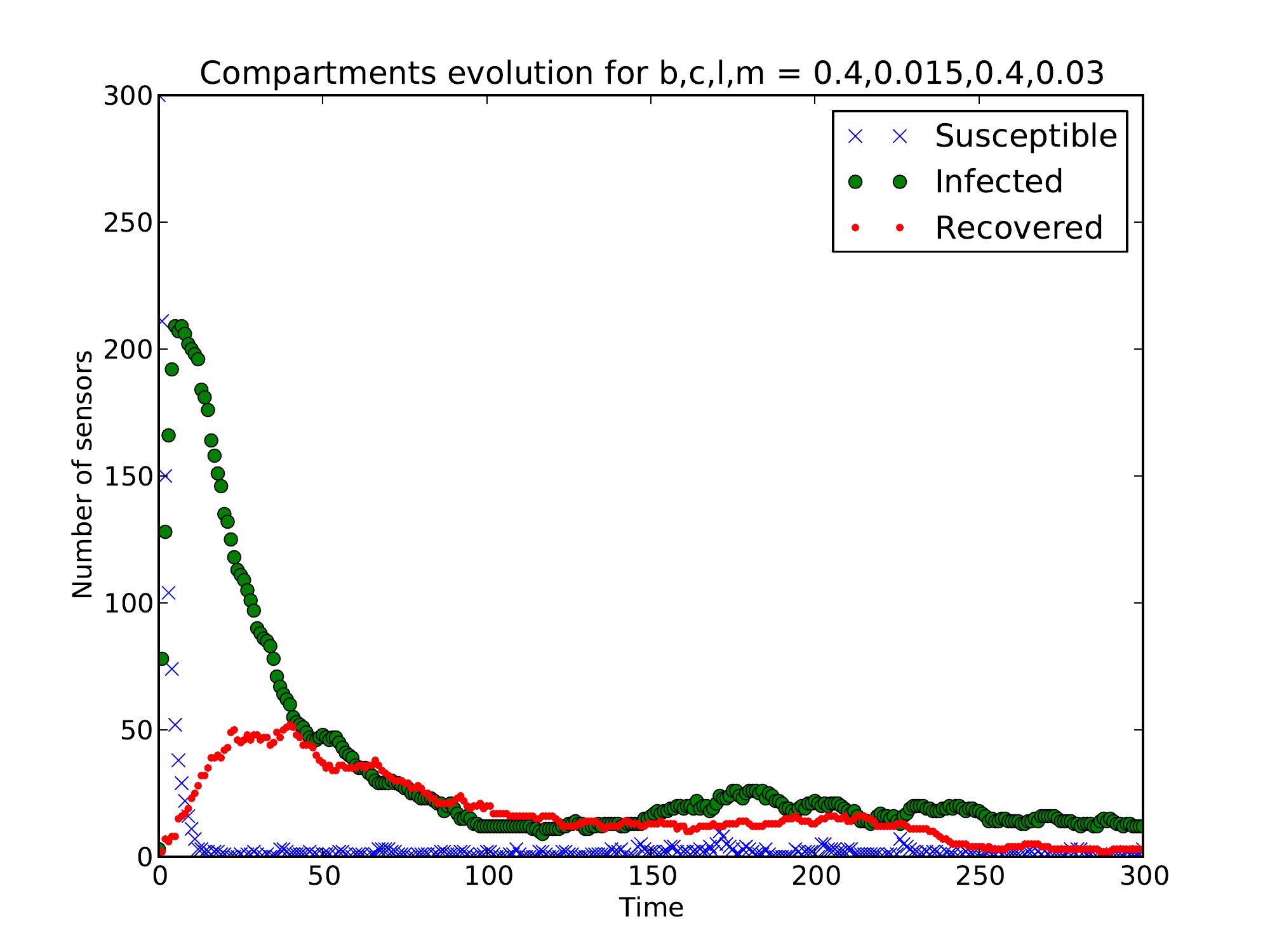}}\quad
  \subfigure[$R_0 = 255.81$]{\includegraphics[scale=.35]{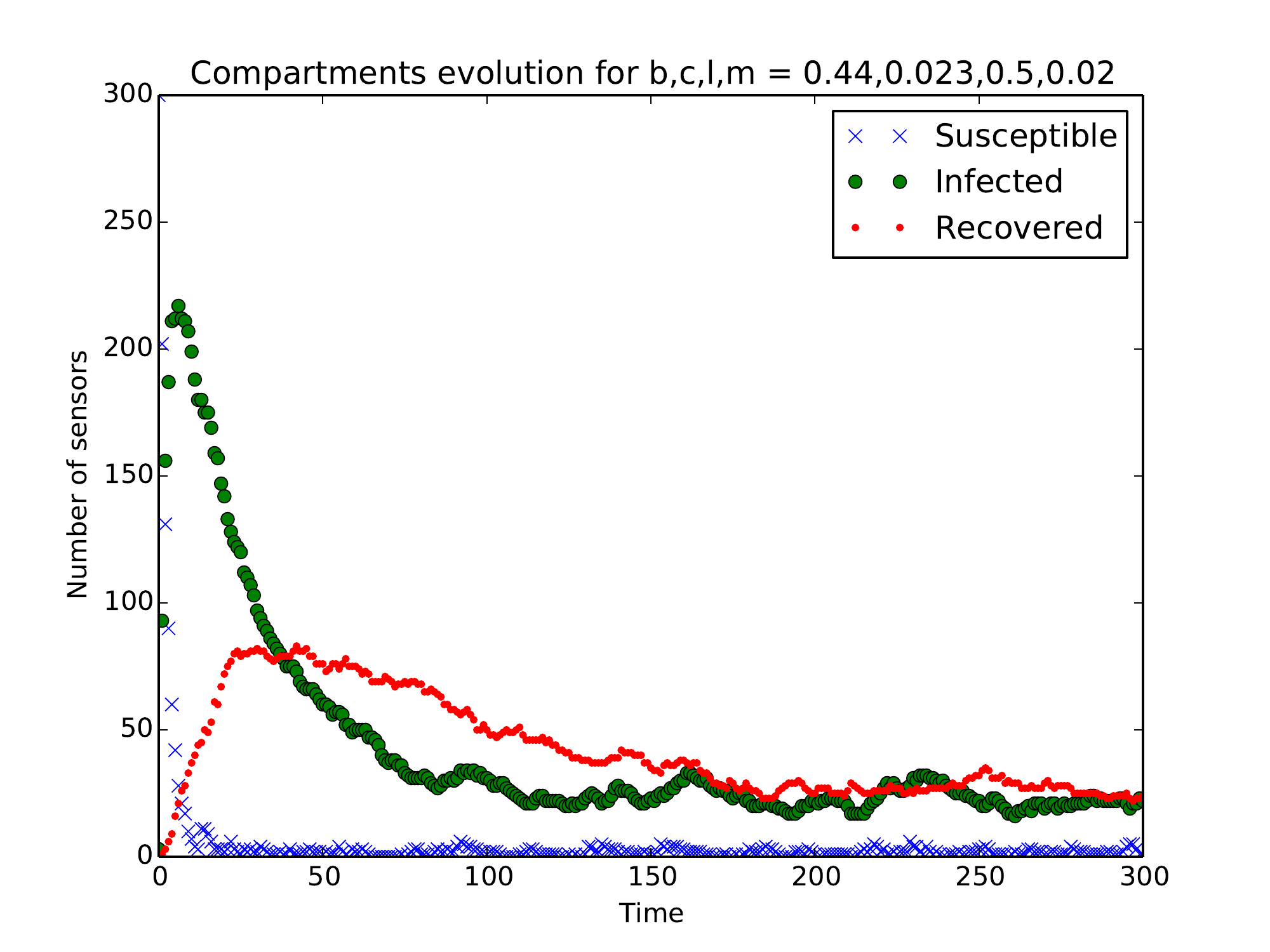}}\\
  \subfigure[$R_0 = 0.29$]{\includegraphics[scale=.35]{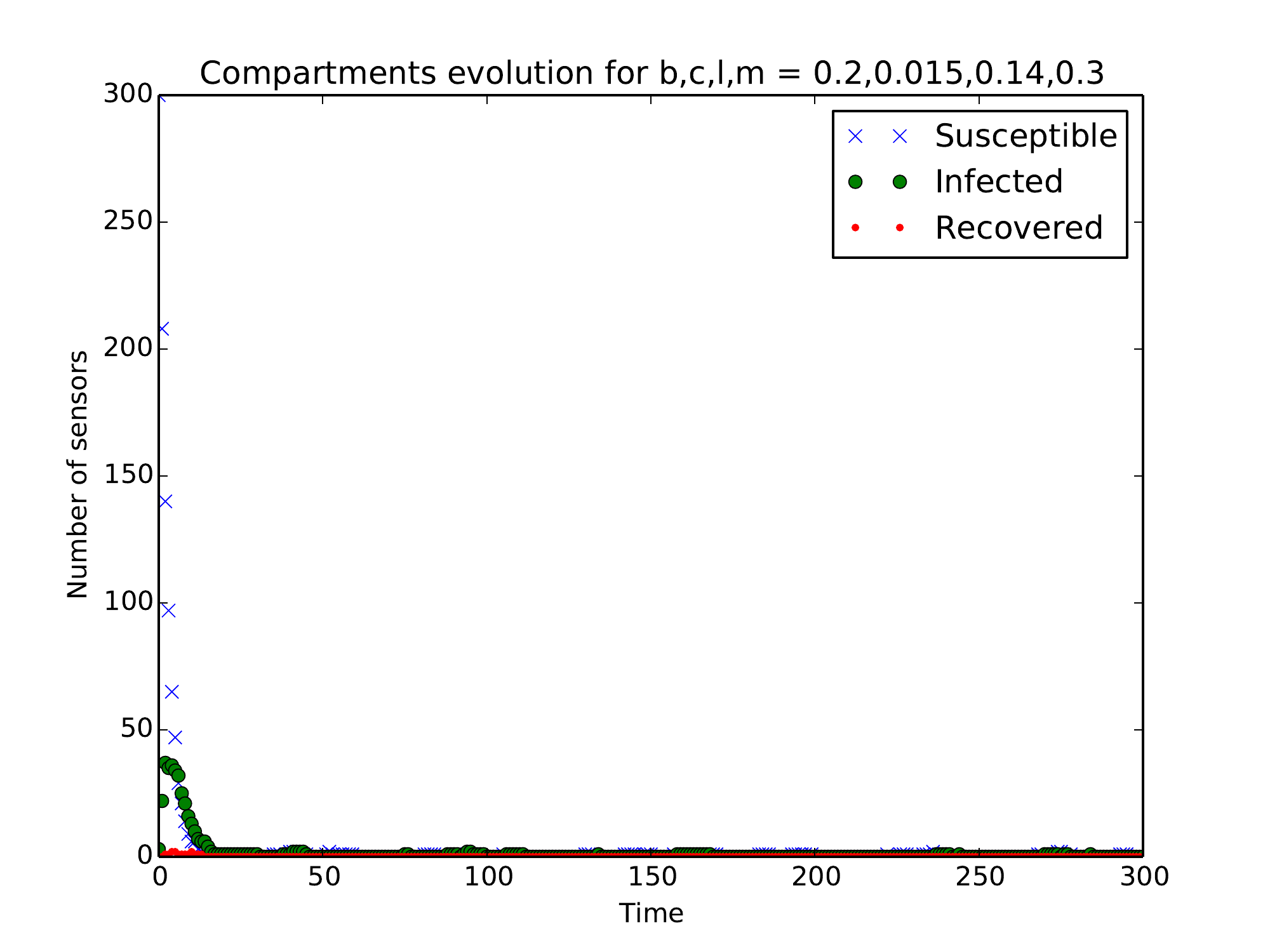}}\quad
  \subfigure[$R_0 = 0.65$]{\includegraphics[scale=.35]{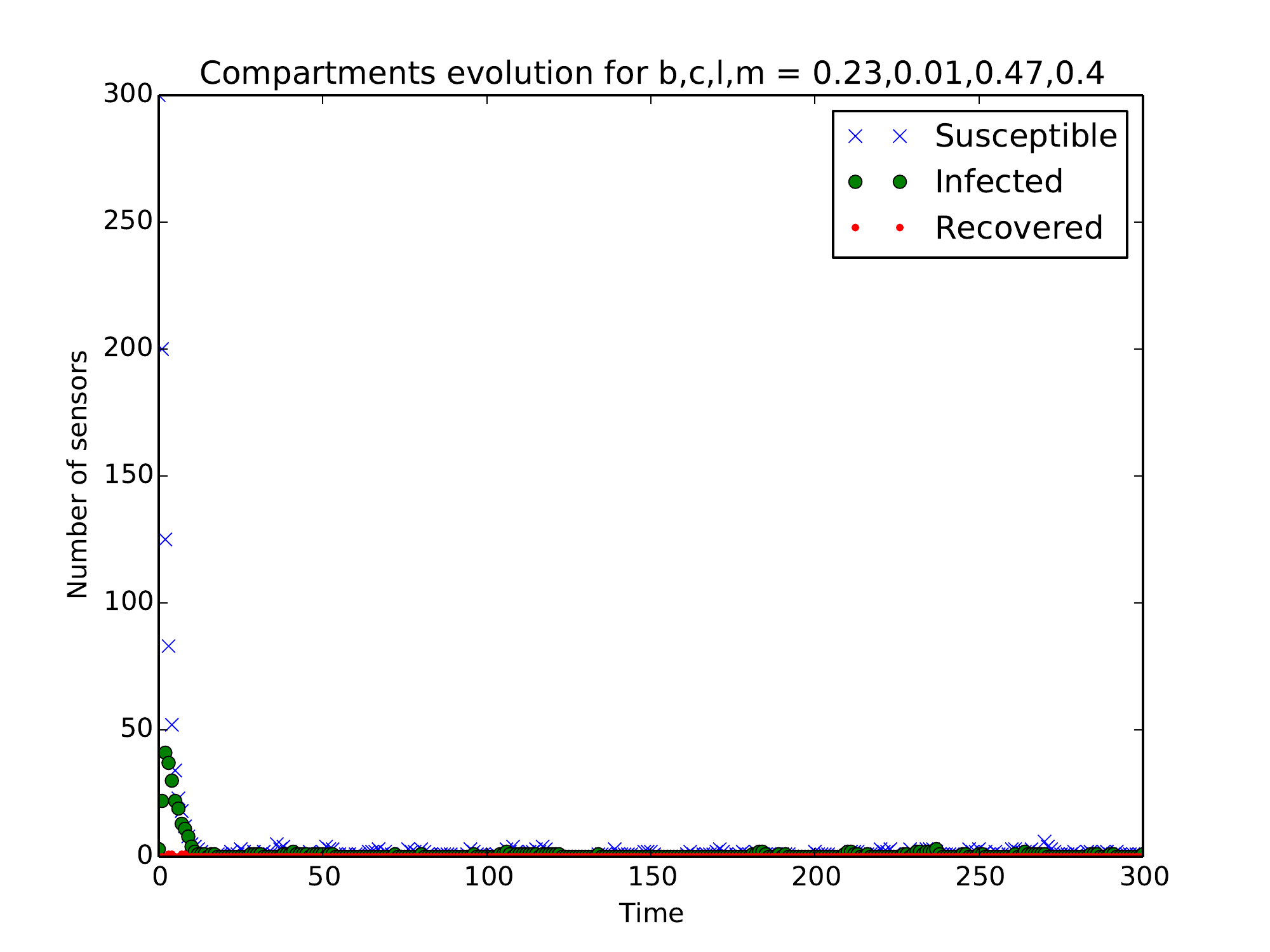}}\\
\end{center}
\caption{Simulation of SIR model with birth and death rates and various $R_0$}
\label{simul}
\end{figure}

\subsection{Networks simulation using Python}

In this second set of experiments, we show that the 
time period of the 
presence of the information can be extended in a wireless sensor network
simulated with Python, and when satisfying Proposition~\ref{laprop}.

\begin{figure}[h]
\centering
\subfigure[$R_0=84.69$]{\label{simulnouvelles1}\includegraphics[scale=0.23]{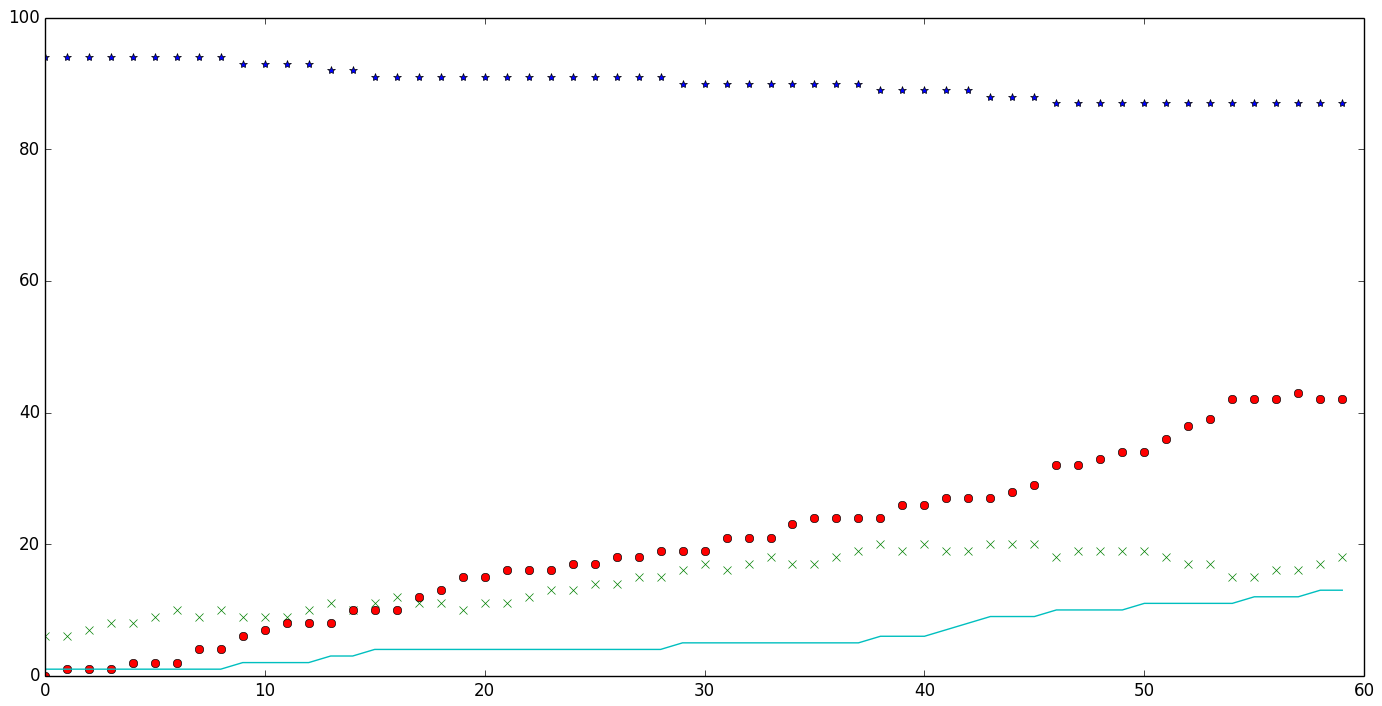}}
\subfigure[$R_0=0.06$]{\label{simulnouvelles2}\includegraphics[scale=0.23]{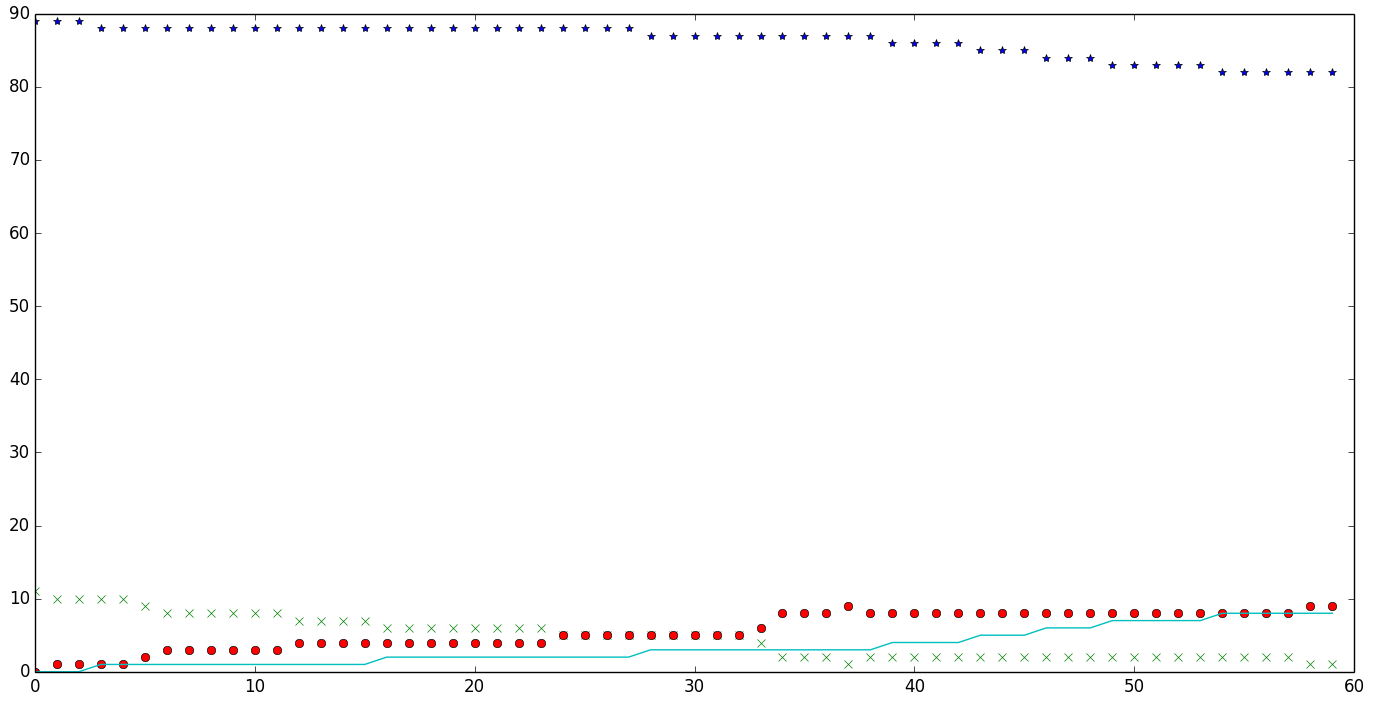}}
\caption{Simulation of a wireless sensor network} 
\end{figure}

We have firstly deployed $N=100$ sensors, all belonging in the susceptible compartment, and
with respect to the algorithm detailed in the previous section.
In the initial condition, each sensor has a probability of 10\% to have detected
an intrusion (this is the information). At each time unit, an average of $lN$ new sensors
are awaken. For each informed sensor and for each of its susceptible neighbor, the data is
sent with a probability $bI$. The death rate of each sensor is set to $m$ (each awaken sensor 
has the probability $m$ to empty its battery during the considered time unit), while 
each informed sensor has a probability $c$ to loose the information (to move in the R 
compartment). The whole network is observed during 60 time units.

We have firstly set $l=0.017$, $m=0.0018$, $c=0.035$, and $b=0.33$, which
leads to $R_0=84.69$, and to the situation depicted in Figure~\ref{simulnouvelles1}.
In this experiment, * symbols have been used for the susceptible sensors, $\times$ for the
informed ones, a circle is for the recovered ones, while the straight line counts the number
of dead sensors. A second set of parameters has led to $R_0=0.06$, and to the
situation described in Figure~\ref{simulnouvelles2}.

\subsection{100 experiments with random parameters}

We have then launched the previous simulator 100 times with random parameters. 
At each simulation, probability $l$ is randomly picked in the interval [0,0.2[,
$m$ is chosen in [0,0.01[, $c$ is picked in [0,0.1[, while $b$ is in [0,0.033[,
in order to be close to a real situation while having $R_0<1$ and $R_0>1$ 
both represented. During these 100 experiments, we have obtained 39 times 
the situation $R_0<1$ with an average of 0.34 (and 61 times the situation $R_0>1$,
16.05 of average).

We found an average number of informed sensors equal to 15.50 in the first situation,
while it is the double in the second one (33.12 informed sensors in average).
In 7 of the 39 simulations with $R_0<1$ (17.95\%), the number of informed sensors has 
reached 0, while the information has disappeared 2 times during the 61 other simulations
(3.27\%). The minimum of informed sensors is attained at the 35-th 
time unit (in average) in the first situation, while we reach it earlier in 
the second one (31-th time unit).

To sum up, the information has disappeared in 3.27\% of the simulations when $R_0>1$, while
it has been lost in 17.95\% of the cases in the second situation.

\section{Conclusion}
\label{CONC}


This paper presented an efficient technique that uses epidemic domain models in the context of data survival in unattended WSNs. We studied two models (SIR and SIS) that can ensure the survivability of the datum in presence of different types of attacks. 
We showed that our method is well adapted to UWSN scenarios. 
In a second step, we proposed and analyzed
an efficient distributed algorithm to tackle the problem of data survivability. 
In future work, the authors' intention is to take into account the possibility of aggregation layers in the wireless sensor networks~\cite{bgm11:ij}: aggregators could transfer only one alert signal for all their neighborhood.

\bibliographystyle{plain}
\bibliography{biblio}

\end{document}